\documentclass[final,onefignum,onetabnum]{siuro250211}

\usepackage{amssymb,latexsym,amsmath}     
\usepackage{epsf}
\usepackage{array}
\usepackage{amssymb,amscd,comment}
\usepackage{float}
\usepackage{enumerate}
\usepackage{geometry}
\usepackage{relsize}%
\usepackage{calc}
\usepackage[mathscr]{eucal}
\usepackage{hhline}
\usepackage{tikz}
\usetikzlibrary{positioning}
\usepackage{enumitem}
\setlist[enumerate]{leftmargin=.5in}
\setlist[itemize]{leftmargin=.5in}
\usepackage{multirow}
\usepackage{tabularx}
\usepackage{soul, color}
\usepackage{amssymb,amstext,amsgen,amsfonts,mathrsfs,verbatim,url,hyperref,mathdots, mathtools}
\usepackage{enumitem}
\usepackage{pgfplots}
\pgfplotsset{compat=1.18}
\usepackage{mathrsfs}
\usepackage{arydshln}
\usepackage{algorithm}
\usepackage{algpseudocode}

\ifpdf
  \DeclareGraphicsExtensions{.eps,.pdf,.png,.jpg}
\else
  \DeclareGraphicsExtensions{.eps}
\fi

\newlist{inparaenum}{enumerate}{2}
\setlist[inparaenum]{nosep}
\setlist[inparaenum,1]{label=\bfseries\alph*.}

\newtheorem{Theorem}[theorem]{Theorem}
\newtheorem{remark}[theorem]{Remark}
\newtheorem{conjecture}[theorem]{Conjecture}
\newtheorem{case}{Case}

\renewcommand{\theenumi}{(\roman{enumi})}
\renewcommand{\theequation}{\thesection.\arabic{equation}}

\def\comment#1{}
\newcommand{\C}{\mathbb{C}}
\newcommand{\R}{\mathbb{R}}
\newcommand{\e}{\mathbf{e}}

\usetikzlibrary{arrows}

\def\invddots{\mathinner{\mskip1mu\raise1pt\vbox{\kern7pt\hbox{.}}\mskip2mu
		\raise4pt\hbox{.}\mskip2mu\raise7pt\hbox{.}\mskip1mu}}

 \author{Mia G. Escobar \thanks{University of Washington - Tacoma, Tacoma, WA (\email{mescob3@uw.edu})} 
 \and Valentin Garcia \thanks{Brown University, Providence, RI (\email{valentin\_garcia@brown.edu})} }

 \dedication{\small\textit{Project advisor: Anastasiia Minenkova}\thanks{University of Hartford, West Hartford, CT (\email{minenkova@hartford.edu})}}

\headers{Early State Exclusion in 7-Qubit Spin Chains}{M. G. Escobar and V. Garcia}
\title{Early State Exclusion in 7-Qubit Spin Chains}

\begin{document}

\maketitle

\begin{abstract}
The existence of infinite families of $N \times N$ Jacobi matrices representing the Hamiltonians of quantum spin chains with and without early state exclusion (ESE) has been shown to exist for any even $N \geq 4$. However, their existence for odd $N \geq 7$ has remained an open problem. In Section 3, we consider a chain of qubits experiencing nearest-neighbor interactions with environmental effects and present infinite families of $7 \times 7$ Jacobi matrices with and without ESE. 

\end{abstract}




\section{Introduction}
Perfect state transfer (PST) in one-dimensional quantum spin chains is widely studied for its applications in quantum information processing. The study of PST was pioneered by Bose \cite{B03}, who analyzed chains with open boundary conditions governed by time-independent Hamiltonians. Subsequently, Kay \cite{Kay10} later showed that the Hamiltonian of a chain of $N + 1$ qubits experiencing nearest-neighbor couplings and environmental effects (see Figure \ref{fig:nearestneighbor}) has a matrix representation of the form of a (symmetric) Jacobi matrix $J$ of order $N + 1$. That is,
\begin{equation}
\label{1.1}
J =
\begin{bmatrix}
a_1 & b_1 & & 0\\
b_1 & a_2 & \ddots & \\
& \ddots  & \ddots & b_N \\
0 & & b_N & a_{N + 1}
\end{bmatrix} 
\quad
\text{for $a_i \in \R$ and $b_i > 0$.}
\end{equation}
\par High fidelity quantum state transfer remains a central problem in quantum computation, with potential to transmit information across quantum systems more efficiently (see \cite{B25}). Here, we focus our attention in transporting states among linear chains from the first to last qubit via the time-evolution operator $e^{-iJt}$. Using the canonical basis $\{\e_k\}_{k = 0}^N$ of $\C^{N + 1}$ with $\e_k$ representing the excited state of the $(k + 1)$th qubit and the convention $\e_0 = [1, 0, \dots, 0]^\top$, we consider the dynamical evolution of $e^{-iJt}\e_0$. Under certain conditions on $J$, there exists some time $T > 0$ for which the probability of finding an excited state in the last qubit is one. This brings us to our first formal definition. 
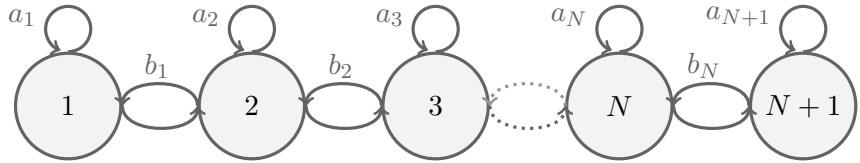
\begin{figure}[ht]
    \begin{center}
    \begin{tikzpicture}[roundnode/.style={circle, draw=gray! 150!, fill=gray!10, very thick, minimum size=14mm},scale=1.5]
        \node[roundnode] (1) {$1$};
        \node[roundnode] (2) [right=of 1]{$2$};
        \node[roundnode] (3) [right=of 2]{$3$};
        \node[roundnode] (N) [right=of 3]{$N$};
        \node[roundnode] (N+1) [right=of N]{$N+1$};
        \draw[->, gray! 150, very thick] (1.east) to[out=-90,in=-90] (2.west) node [above left=9pt]{$b_1$};
        \draw[->, gray! 150, very thick] (2.west) to[out=90,in=90] (1.east) ;
        \draw[->, gray! 150, very thick] (2.east) to[out=-90,in=-90] (3.west)node [above left=9pt]{$b_2$};
        \draw[->, gray! 150, very thick] (3.west) to[out=90,in=90] (2.east) ;
        \draw[->, gray! 150, very thick] (N.east) to[out=-90,in=-90] (N+1.west)node [above left=9pt]{$b_N$};
        \draw[->, gray! 150, very thick] (N+1.west) to[out=90,in=90] (N.east) ;
        \draw[->, dotted, gray! 150, very thick] (3.east) to[out=-90,in=-90] (N.west);
        \draw[->, dotted, gray, very thick] (N.west) to[out=90,in=90] (3.east) ;
        \draw[->, gray! 150, very thick] (1.north)arc(-90:250:0.2) node[above left= 7pt]     {$a_1$}; 
        \draw[->, gray! 150, very thick] (2.north)arc(-90:250:0.2)node[above left= 7pt]    {$a_2$}; 
        \draw[->, gray! 150, very thick] (3.north)arc(-90:250:0.2)node[above left= 7pt]     {$a_3$}; 
        \draw[->, gray! 150, very thick] (N.north)arc(-90:250:0.2)node[above left= 7pt]    {$a_N$}; 
        \draw[->, gray! 150, very thick] (N+1.north)arc(-90:250:0.2)node[above left= 7pt]    {$a_{N+1}$}; 
        \coordinate (A) at (1,1);
    \end{tikzpicture}
    \end{center}
    \caption{Prototypical model of a quantum spin chain with $N + 1$ qubits experiencing nearest-neighbor interactions and environmental effects.}
    \label{fig:nearestneighbor}
\end{figure}
\begin{definition}
\emph{A Jacobi matrix $J$ realizes \textbf{perfect state transfer} (PST) between the end-vertices of the weighted path in Figure \ref{fig:nearestneighbor} at time $T > 0$ if
\begin{equation}
\label{1.2}
e^{-iJT}\e_0 = e^{i\phi}\e_N
\end{equation}
for some phase $\phi \in \R$.}
\end{definition}
\begin{remark}
\emph{From here on, when we state that a matrix $J$ realizes PST, this refers exclusively to PST between the terminal vertices of its associated weighted path graph.}
\end{remark}
\begin{figure}[ht]
    \centering
    \includegraphics[width=0.55\linewidth]{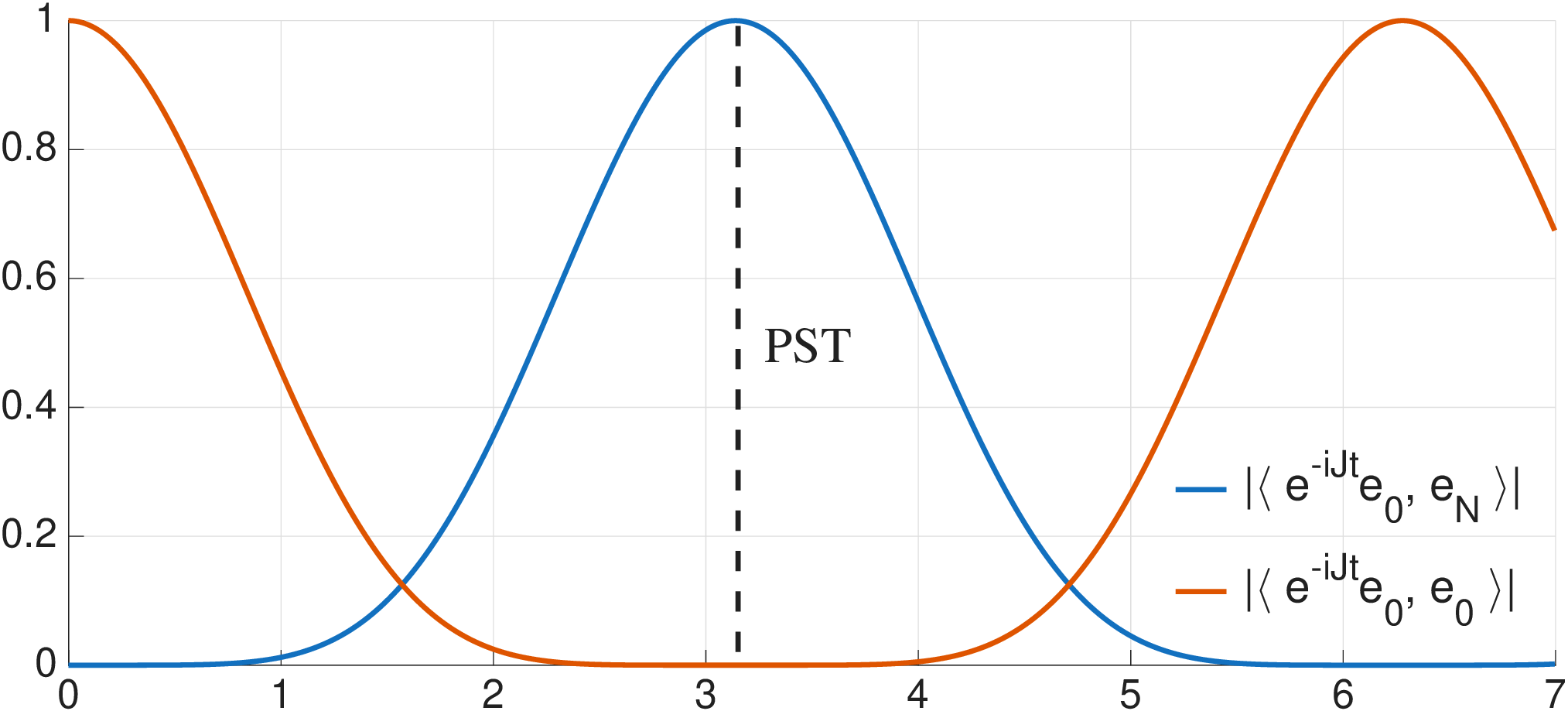}
    \caption{Instance of PST for $J$ with symmetric spectrum $\{0, \pm 1, \pm 2, \pm 3\}$.}
    \label{fig:enter-label}
\end{figure}
\par Kay \cite{Kay10} established that a necessary and sufficient condition for $J$ to realize PST is that $J$ be persymmetric and its ordered eigenvalues $\mu_0 < \mu_1 < \cdots < \mu_N$ satisfy the relation
\begin{equation}
\label{1.3}
\mu_{k + 1} - \mu_k = \frac{(2n_k + 1)\pi}{T} \quad \text{for each $0 \leq k \leq N - 1$,}
\end{equation}
where each $n_k$ is a nonnegative integer. We consider a special case where $J$ has symmetric spectrum in section 3, meaning that for $\{\lambda_k\}_{k = 0}^N$, $\lambda_j = -\lambda_{N - j}$ for all $1 \leq j \leq N - 1$. If we assume that the positive eigenvalues of $J$ are coprime, then \eqref{1.3} implies that $J$ first realizes PST at time $T = \pi$.
\par In some Jacobi matrices exhibiting PST, the probability amplitude of the first qubit has been observed to vanish before the first instance of PST. This phenomenon is referred to as early state exclusion.



\begin{definition}
\label{ese}
\emph{Let $J$ be a Jacobi matrix realizing PST at earliest time $T > 0$. If there exists some $\tau \in (0, T)$ such that
\begin{equation}
\label{1.4}
    \langle e^{-iJ\tau}\e_0, \e_0 \rangle = 0,
\end{equation}
then we say that $J$ exhibits \textbf{early state exclusion} (ESE) at time $\tau$.}
\end{definition}
\begin{figure}[H]
    \centering
    \includegraphics[width=0.5\linewidth]{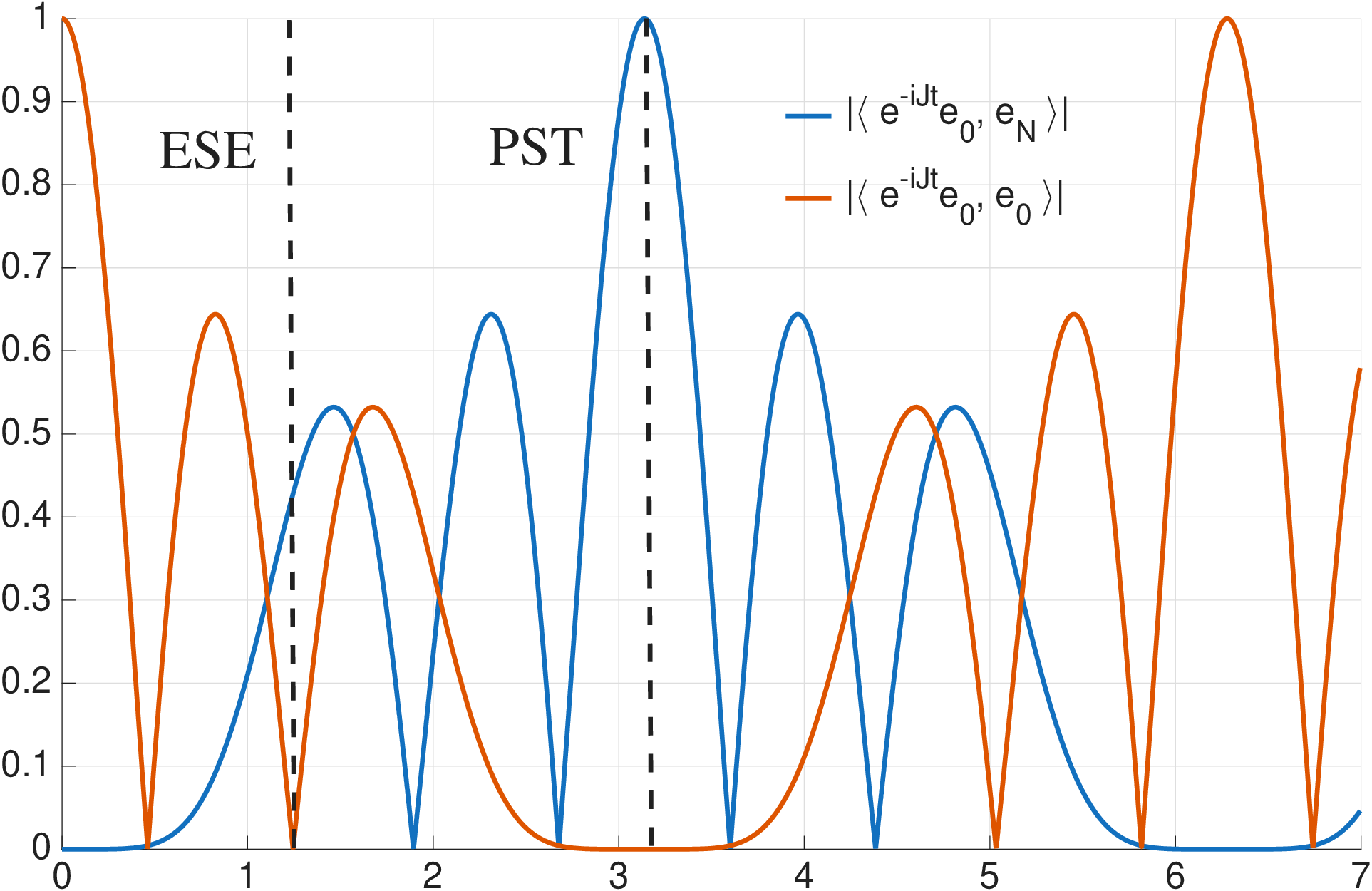}
    \caption{Instance of ESE for $J$ with symmetric spectrum $\{0, \pm 3, \pm 4, \pm 5\}$.}
    \label{fig:enter-label}
\end{figure}
The authors of \cite{ESE} showed that for any even number of qubits, an infinite family of corresponding Jacobi matrices that exhibit early state exclusion exists. For five qubits, infinite families of such matrices--both with and without ESE--have been identified \cite{EM25}. However, the question remains open for higher odd numbers of qubits \cite{B25}. In section 3, we will show that there exist infinite families of Jacobi matrices of order $7$ realizing PST that have ESE and others that have no ESE.
\section{Reconstructing Persymmetric Jacobi Matrices}
\par The Hochstadt uniqueness theorem \cite{H74} states that any persymmetric Jacobi matrix $J$ of order $N + 1$ can be uniquely reconstructed from its spectrum $\{\lambda_j\}_{j = 0}^N$. After introducing some background in orthogonal polynomials, we conclude this section by explicitly computing the entries of a $7 \times 7$ persymmetric Jacobi matrix $J$ with a symmetric spectrum.
\par Let $p_i$ be the characteristic polynomial of the $i$th principal submatrix of the persymmetric Jacobi matrix $J$ defined in \eqref{1.1}. Then the sequence of polynomials $\{p_i\}_{i = 1}^{N + 1}$ satisfies the three-term recurrence relation
\[
p_i(t) = (t - a_i)p_{i - 1}(t) - b_{i - 1}^2p_{i - 2}(t),
\]
where $p_0 \equiv 1$ and $p_{-1} \equiv 0$. We define an inner product $\langle \cdot, \cdot \rangle_w$ on the vector space of polynomials with real coefficients of degree at most $N + 1$ by
\[
\langle f, g \rangle_w = \sum_{j = 0}^{N} f(\lambda_j)g(\lambda_j)w_j,
\]
where the weights $w_j$ are given by
\[
w_j = \prod_{\substack{i = 0 \\ i \neq j}}^{N} \frac{1}{|\lambda_j - \lambda_i|}.
\]
Note that $w_j$ is well-defined as $\lambda_i \neq \lambda_j$ whenever $i \neq j$ (see \cite{H67}). It follows from the recursion relation that $\{p_i\}_{i = 1}^{N + 1}$ is an orthogonal sequence with respect to $\langle \cdot, \cdot \rangle_w$. Hence, the entries of $J$ can be reconstructed via
\[
a_i = \frac{\langle tp_{i - 1}, p_{i - 1} \rangle_w}{\|p_{i - 1}\|_w^2} \quad \text{and} \quad b_i = \frac{\|p_i\|_w}{\|p_{i - 1}\|_w},
\]
where $\|\cdot\|_w$ denotes the norm induced by the inner product. 
\subsection{Symmetric Spectra}
As previously mentioned, we consider Jacobi matrices with symmetric spectra. Under these conditions, a significant simplification of $J$ occurs where $a_1 = \cdots = a_{N + 1} = 0$. This is due to the fact that $\langle tp_{i - 1}, p_{i - 1} \rangle_w = 0$. The algorithm from \cite{BG} can thus be adapted for the case of symmetric spectra in Algorithm \ref{algorithm}.
\begin{algorithm}[H]
\label{algorithm}
\caption{Reconstruction of a Persymmetric Jacobi Matrix with Symmetric Spectrum}
    \begin{algorithmic}[1]
        \Require Prescribed symmetric spectrum $\{\lambda_j\}_{j = 0}^{N}$
        \For{$i = 0, 1, \dots, N$}
        \Comment{Weights}
            \State {$w_i \gets 1/\prod_{j = 0, j \neq i}^{n} |\lambda_j - \lambda_i|$};
        \EndFor
        \Function{$\|f\|_w$}{}\Comment{Polynomial Norm}
            \State \Return {$\sqrt{\sum_{j = 0}^N f^2(\lambda_j)w_j}$};
        \EndFunction
        \State $p_0(t) \gets 1$; \Comment{Initial Conditions}
        \State $p_1(t) \gets t$;  
        \State $b_1 \gets \|p_1\|_w/\|p_0\|_w$;
        \For{$i = 2, \dots, N$} \Comment{Three-term recurrence relation}
            \State $p_i(t) \gets tp_{i - 1}(t) - b_{i - 1}^2p_{i - 2}(t)$;
                \State $b_i \gets \|p_i\|_w/\|p_{i - 1}\|_w$;
            \EndFor
        \State \Return $[b_1, \dots,b_N]$ \Comment{Off-diagonal elements of $J$}
\end{algorithmic} 
\end{algorithm} 
Matching the off-diagonal elements of \eqref{1.1} with the output $[b_1, \dots, b_N]$ from above thereby reconstructs the persymmetric Jacobi matrix $J$.
\subsection{\texorpdfstring{$7 \times 7$}{7 x 7} Persymmetric Jacobi Matrices with Symmetric Spectrum}

For any positive real numbers $0 < x < y < z$, the general form of a $7 \times 7$ persymmetric Jacobi matrix with symmetric spectrum $\{0, \pm x, \pm y, \pm z\}$ is given by
\[
J =
\begin{bmatrix}
0 & b_1 & & & & & 0\\
b_1 & 0 & b_2& & &\\
& b_2 & 0 & b_3 \\
& & b_3 & 0 & b_3 \\
& & & b_3 & 0 & b_2 \\
& & & & b_2 & 0 & b_1 \\
0 & & & & & b_1 & 0
\end{bmatrix},
\]
where the off-diagonal elements are
\[
b_1 = \frac{xz}{\sqrt{x^2-y^2+z^2}}, \quad b_2 = \sqrt{\frac{(y^2-x^2)(z^2-y^2)}{x^2-y^2+z^2}}, \quad \text{and} \quad b_3 = \sqrt{\frac{x^2-y^2+z^2}{2}}.
\]
To compute the inner product in \eqref{1.4}, we label the nonnegative eigenvalues of $J$ as $\lambda_0 = 0$, $\lambda_1 = x$, $\lambda_2 = y$, and $\lambda_3 = z$. Applying Sylvester's formula yields
\[
e^{-iJt} = \sum_{j = 0}^6 e^{-i\lambda_j t} J_j,
\]
where $\lambda_k = -\lambda_{6 - k}$ for $k = 1, 2, 3$ and $J_j$ is the Frobenius covariant
\[
J_j \equiv \prod_{\substack{i = 0 \\ i\neq j}}^6 \frac{J- \lambda_j I}{\lambda_j - \lambda_i}
\]
with $I$ acting as the $7 \times 7$ identity matrix. We consider the inner products
\begin{align*}
    &\langle (J^2 - x^2I)(J^2 - y^2I)(J^2 - z^2I)\e_0, \e_0 \rangle = \frac{(y^2 - x^2)(z^2 - y^2)}{2y^2(x^2 - y^2 + z^2)}\\
   &\langle (J^2 - y^2I)(J^2 - z^2I)(J^2\cos{xt} - xJ \sin{xt})\e_0, \e_0 \rangle = \frac{z^2(z^2 - y^2)\cos{xt}}{2(z^2 - x^2)(x^2 - y^2 + z^2)}\\
    &\langle (J^2 - x^2I)(J^2 - z^2I) (J^2\cos{yt}-yJ \sin{yt})\e_0, \e_0 \rangle = \frac{x^2z^2\cos{yt}}{2y^2(x^2 - y^2 + z^2)}\\
   & \langle (J^2 - x^2I)(J^2 - y^2I)(J^2\cos{zt} -yJ \sin{zt})\e_0, \e_0 \rangle = \frac{x^2(y^2 - x^2)\cos{zt}}{2(z^2 - x^2)(x^2 - y^2 + z^2)},
\end{align*}
which by linearity yield
\begin{align}
\label{2.1}
    \langle e^{-iJt}\e_0, \e_0 \rangle &= \frac{1}{2y^2(z^2 - x^2)(x^2-y^2+z^2)}\Big((y^2 - x^2)(z^2 - x^2)(z^2 - y^2)\\
        \nonumber 
        &\quad + y^2z^2(z^2 - y^2)\cos{xt} + x^2z^2(z^2 - x^2)\cos{yt} + x^2y^2(y^2 - x^2)\cos {zt}\Big).
\end{align}
We can express this equivalently as
\begin{equation}
\label{2.2}
A(t) = \langle e^{-iJt}\e_0, \e_0 \rangle = c_0 +  \sum_{k = 1}^{3} c_k\cos{\lambda_kt},
\end{equation}
where the coefficients above are 
\[
c_k = \prod_{\substack{i = 0 \\ i \neq k}}^3 \frac{1}{|\lambda_k^2 - \lambda_i^2|} \Bigg/\sum_{j = 0}^{3} \prod_{\substack{i = 0 \\ i \neq j}}^3 \frac{1}{|\lambda_j^2 - \lambda_i^2|}.
\]

\section{Infinite Families of \texorpdfstring{$7 \times 7$}{7 x 7} Jacobi Matrices With and Without ESE}

In this section, we will prove the existence of an infinite family of $7 \times 7$ Jacobi matrices without ESE, and an infinite family of $7 \times 7$ Jacobi matrices with ESE. We focus specifically on those with symmetric spectra. 

\begin{remark}
\label{3.1}
\emph{For a spectrum of coprime integers, observe that $A$ has a root at $t = \pi$ with multiplicity $6$. In addition, because of the swapping from cosines to sines during differentiation, the value of the odd derivatives will also always be zero at integer multiples of $\pi$. So, starting with $A^{(6)}(\pi)$, we note that, at least until the 12th derivative, the signs of values of even order derivatives at $\pi$ alternate:
    \begin{align*}
    A^{(6)}(\pi) &= \lambda_1^6c_1 - \lambda^6_2c_2 + \lambda^6_3c_3 > 0\\
    A^{(8)}(\pi) &= -(\lambda_1^8c_1 - \lambda^8_2c_2 + \lambda^8_3c_3) < 0 \\
    A^{(10)}(\pi) &= \lambda_1^{10}c_1 - \lambda^{10}_2c_2 + \lambda^{10}_3c_3>0\\
    A^{(12)}(\pi) &= -(\lambda_1^{12}c_1 - \lambda^{12}_2c_2 + \lambda^{12}_3c_3) < 0.
    \end{align*}
}
\end{remark}

\subsection{Infinite Family of \texorpdfstring{$7 \times 7$}{7 x 7} Jacobi Matrices without ESE}
\begin{Theorem}
\label{Theorem 3.2}
Suppose that $J$ is a Jacobi matrix of order $7$ realizing PST with symmetric spectrum 
\[
\{0, \pm 1, \pm 2m, \pm (2m + 1)\}
\]
for some integer $m \geq 1$. Then $J$ \textbf{does not} have ESE.
\end{Theorem}
\vspace{0.25cm}
\begin{proof}
The earliest time for which $J$ realizes PST is at $T = \pi$, so we can write $A(t)$ as
\begin{align*}
    A(t) &= \frac{1}{16m^2(1 + m)}\big(m(1 + 2m)(1 + 4m)(\cos{t} + 1) + m(2m - 1)(\cos\big((2m+1)t\big) + 1)\\ 
    &\qquad +(1 + m)(1 + 2m)(\cos(2mt) - 1)\big).
\end{align*}
We study two functions $f$ and $g$ such that $A=\frac{1}{16m^2(1 + m)} \big(f-g\big)$, where
\begin{align*}
    &f(t)=m(1+2m)(1+4m)(\cos{t}+1)+m(2m-1)(\cos\big((2m+1)t\big)+1), \text{ and}\\ 
    &g(t) =-(1+m)(1+2m)(\cos(2mt)-1).
\end{align*}
Note that the first derivatives of $f$ and $g$ take the form
\begin{align*}
    f'(t) &= -m(1+2m)(1+4m)\sin{t}-m(2m-1)(1+2m)(\sin\big((2m+1)t\big)),\text{ and} \\
    g'(t) &= (1+m)(1+2m)(2m)\sin(2mt).
\end{align*}
The function $g$ has extrema on $[0,\pi]$ at values $t=\frac{k\pi}{2m}$ where $k=1,2,...2m$. The maximums of $g$ occur when $k$ is odd, which gives $g\big(\frac{k\pi}{2m} \big) = (1+m)(1+2m)(2m)$. For $t\leq \frac{(2m-1)\pi}{2m}-\frac{1}{\sqrt{m}}$, $$f(t)> (1+m)(1+2m)(2m) \geq g(t),$$ which implies $A(t)>0$.

Now, we shift our focus to the interval $\left(\frac{(2m-1)\pi}{2m}-\frac{1}{\sqrt{m}}, \pi- \frac{2\pi}{2m+1} \right)$. If we know that $f>g$ on the endpoints of our interval, then it suffices to show that $g'(t)<f'(t)$ on this interval to prove $A(t)>0$, because $f$ and $g$ would have no intersection. 
\begin{figure}[H]
   \centering
   \includegraphics[width=0.5\linewidth]{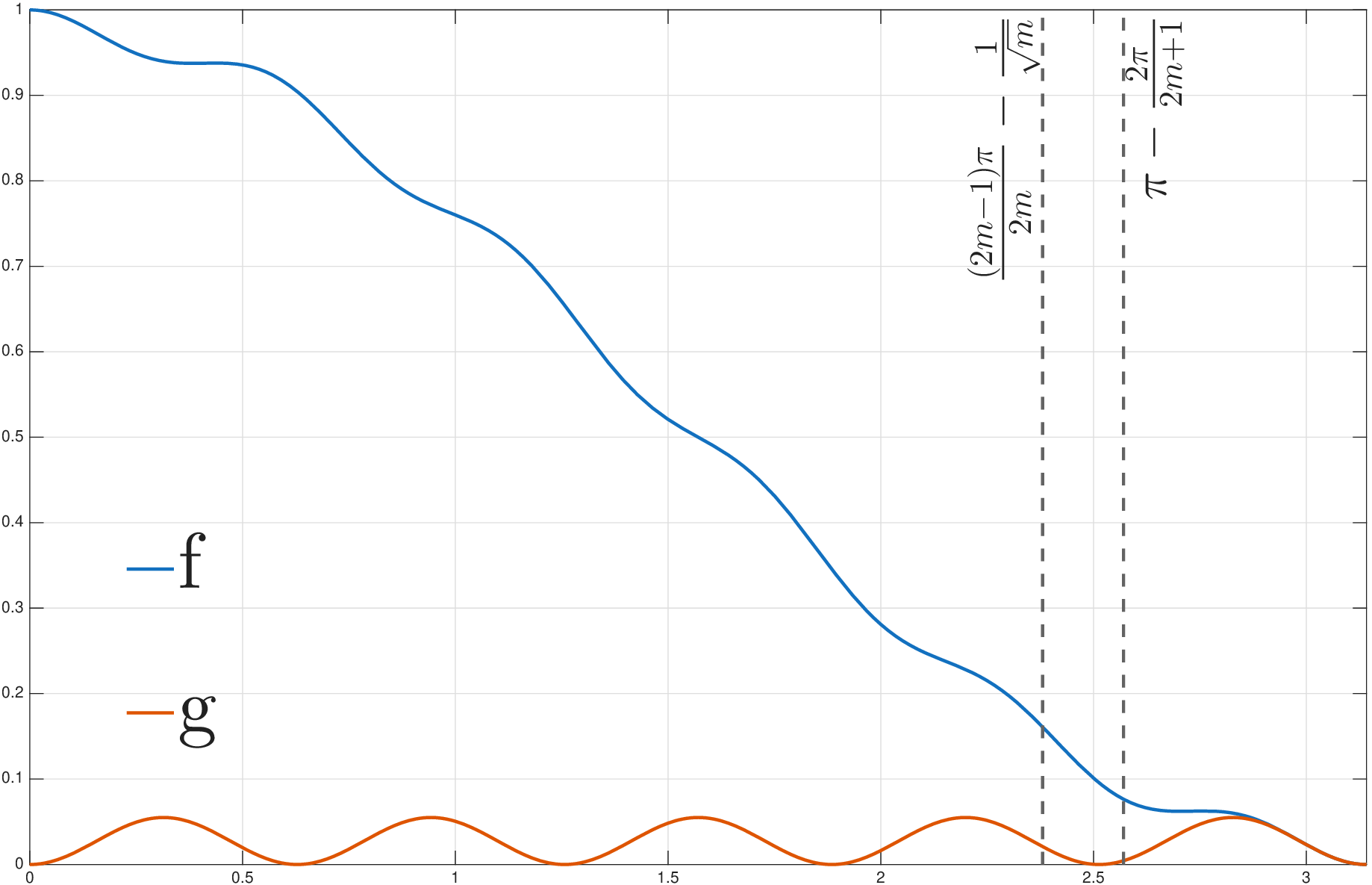}
   \caption{Comparison of normalized $f$ and $g$ for $m=5$.}
   \label{fig:enter-label}
\end{figure}
The graph of $f$ does not cross $y=(1+m)(1+2m)(2m)$ until after $\frac{(2m-1)\pi}{2m}-\frac{1}{\sqrt{m}}$, so we know that $f\left(\frac{(2m-1)\pi}{2m}-\frac{1}{\sqrt{m}}\right)>g\left(\frac{(2m-1)\pi}{2m}-\frac{1}{\sqrt{m}}\right)$. We note the values of $f(t)$ and $g(t)$ at $t=\pi- \frac{2\pi}{2m+1}$ as 
\begin{equation}
  \label{f(pi-2pi/2m+1)}  f\left(\pi- \frac{2\pi}{2m+1}\right)= m(1+2m)(1+4m)\left(\cos \left(\pi- \frac{2\pi}{2m+1}\right)+1\right)\\
  \end{equation}
\begin{equation}
  \label{g(pi-2pi/2m+1)}  g\left(\pi- \frac{2\pi}{2m+1}\right) = -(1+m)(1+2m)\cos \left( \frac{4m\pi}{2m+1}\right)
\end{equation}
Studying \eqref{f(pi-2pi/2m+1)} and \eqref{g(pi-2pi/2m+1)} as functions of $m$, the least value $f\big(\pi- \frac{2\pi}{2m+1}\big)$ reaches is when $m=1$, specifically $f\big(\pi- \frac{2\pi}{3}\big) = 22.5$, after which for $m\geq1$ it continues to increase. The greatest value $g\big(\pi- \frac{2\pi}{2m+1}\big)$ reaches also occurs at $m=1$ when $g\big(\pi- \frac{2\pi}{3}\big)=3$, after which $g\big(\pi- \frac{2\pi}{2m+1}\big)$ decreases. So, $f\big(\pi- \frac{2\pi}{2m+1}\big) >g\big(\pi- \frac{2\pi}{2m+1}\big)$ for all $m$.

Now that we have established for both end points of $\big(\frac{(2m-1)\pi}{2m}-\frac{1}{\sqrt{m}}, \pi- \frac{2\pi}{2m+1} \big)$ that $f>g$, we want to show 
\begin{equation}
\begin{aligned}
\label{difference}
 (f-g)'(t)=-&m(1+2m)(1+4m)\sin{t}-m(2m-1)(1+2m)(\sin\big((2m+1)t\big))-\\
  &(1+m)(1+2m)(2m)\sin(2mt)<0.
\end{aligned}
\end{equation}
We re-parameterize \eqref{difference} by defining $t=\pi-s$ so that we can simplify the presentation of algebraic manipulations and study $s \in \big(\frac{2\pi}{2m+1}, \frac{(2m-1)\pi}{2m}+\frac{1}{\sqrt{m}} \big)$ instead, which results in the equation 
 \begin{equation}
    \begin{aligned}
    \label{f'-g'}
    (f-g)'(\pi-s)=& -m(1+2m)(1+4m)\sin(\pi-s)-m(4m^2-1)\sin((2m+1)(\pi-s)) -\\&(1+m)(1+2m)(2m)\sin(2m(\pi-s))
  =-m(1+2m)((1+4m)\sin(s)+\\&(2m-1)\sin((2m+1)s)-(2m+2)\sin(2ms)).
\end{aligned}
\end{equation}
First, we prove an auxiliary lemma.
\\
\begin{lemma}
\label{sinelemma}
For $m \geq 1$ the following inequality is true,
\[
-2\sin \frac{s}{2}\leq \sin(2ms) -\sin\big((2m+1)s\big) \leq 2\sin \frac{s}{2}.
\]
\end{lemma}
\begin{proof}
Recall that the Chebyshev polynomials of the first and second kind are defined by
\begin{equation}
\label{chebyshevdefinitions}
T_n(\cos\theta) = \cos(n\theta) \quad \text{and} \quad  U_n(\cos\theta) = \frac{\sin\big((n+1)\theta\big)}{\sin\theta},
\end{equation}
   respectively. We know that $$-1\leq -\cos \big((4m+1)\theta\big)\leq 1,$$ so we can express this as $$-1\leq -T_{4m+1}(\cos\theta)\leq 1.$$ By the properties of Chebyshev polynomials, we can write this as $$-1\leq\frac{1}{2}\left(U_{4m-1}(\cos\theta)-U_{4m+1}(\cos\theta)\right)\leq 1.$$ By \eqref{chebyshevdefinitions}, this is equivalent to $$-1\leq \frac{1}{2}\left( \frac{\sin(4m\theta)}{\sin\theta}-\frac{\sin\big((4m+2)\theta\big)}{\sin\theta} \right)\leq 1.$$ Multiplying both sides by $2\sin\theta$ and substituting for $\theta=\frac{s}{2}$, we are left with $$-2\sin \frac{s}{2} \leq \sin(2ms)-\sin\big((2m+1)s\big) \leq 2\sin \frac{s}{2}.$$
 This concludes the proof of Lemma \ref{sinelemma}. 
\end{proof}

By Lemma \ref{sinelemma}, we know that $\sin(2ms)-\sin\big((2m+1)s\big) \leq 2\sin\frac{s}{2}$, which leaves the upper bound for function \eqref{f'-g'} as
\begin{align*}
    (f-g)'(\pi-s) &\leq -m(1+2m)\left( (1+4m)\sin{s}-3\sin\big((2m+1)s\big)-2(2m+2)\sin \frac{s}{2} \right)\\
   &= -m(1+2m)\left( 2(1+4m)\sin\frac{s}{2}\cos\frac{s}{2}-3\sin\big((2m+1)s\big)-2(2m+2)\sin\frac{s}{2}\right)\\
   &\leq -m(1+2m)\left( 2(1+4m)\sin\frac{s}{2}\cos\frac{s}{2}-2(2m+2)\sin \frac{s}{2}-3\right)\\
   &=  -m(1+2m)\left( (4m+4)\sin \left( \frac{s}{2}\right) \left(-1+\frac{2(1+4m)}{4m+4} \cos\left( \frac{s}{2}\right)\right)-3\right)\\
   &\leq -m(1+2m)\left( (4m+4)\left( \frac{s}{2}\right) \left(-1+\frac{2(1+4m)}{4m+4} \left(1-\frac{s^2}{8}\right)\right)-3 \right).
\end{align*}

From this, we conclude that if
\begin{equation}
\label{lhs}
    (4m+4)\left( \frac{s}{2}\right) \left(-1+\frac{2(1+4m)}{4m+4} \left(1-\frac{s^2}{8}\right)\right) > 3
\end{equation}
is true, then $A'(t)<0$ for $t\in \big(\frac{(2m-1)\pi}{2m}-\frac{1}{\sqrt{m}}, \pi- \frac{2\pi}{2m+1} \big)$. Equivalently, to show that \eqref{difference} is true, it is enough to show that the inequality \eqref{lhs} holds. Simplifying further and plugging in our upper boundary condition for $s$, we find that the left hand side of \eqref{lhs} is also bounded below by
\begin{align*}
   &(4m+4)\left(\frac{\pi}{4m}+\frac{1}{2\sqrt{m}} \right) \left(-1+ \frac{(1+4m)(8(2m+1)^2-4\pi^2)}{16(m+1)(2m+1)^2}\right).
\end{align*}
Studying $$\left(\frac{(1+4m)(8(2m+1)^2-4\pi^2)}{16(m+1)(2m+1)^2}\right)$$ as a function of $m$, we can see that it is increasing for $m\geq1$, and is bounded below by $1.5$ for $m \geq 4$. From here on out we consider $m \geq 4$. Then
\begin{align*}
   (4m+4)\left( \frac{s}{2}\right) &\left(-1+\frac{2(1+4m)}{4m+4} \left(-1+\frac{s^2}{8}\right)\right)\geq (4m+4)\left(\frac{\pi}{4m}+\frac{1}{2\sqrt{m}} \right) \left(\frac{1}{2}\right) \\
   &= (2m+2)\left(\frac{\pi}{2m}+\frac{1}{\sqrt{m}} \right)= \frac{\pi}{2} + \frac{\pi}{m}+\sqrt{m} + \frac{1}{\sqrt{m}}>3.
\end{align*}
Therefore, $A'(t)<0$ and $A(t)>0$ for $m \geq 4$.

\par Now that we have shown $A(t)>0$ on $\big(\frac{(2m-1)\pi}{2m}-\frac{1}{\sqrt{m}}, \pi- \frac{2\pi}{2m+1} \big)$ for $m\geq 4$, we can go case by case for $1\leq m \leq 3$. 
\begin{case}
     \emph{When $m=1$, we have an equidistant spectrum, which was shown in \cite{ESE} to never have ESE.}
\end{case}
\begin{case}
    \emph{When $m=2$,  $A(t)= \frac{1}{64}(27 + 30\cos{t} + 5\cos{4t} + 2\cos{5t})$. We can write this in terms of Chebyshev polynomials of the first kind, substituting $T_1(\cos(t))=x$. So, $$\frac{1}{64}(27 + 30x + 5(8x^4 - 8x^2 + 1) + 2(16x^5 - 20x^3 + 5x)) = \frac{1}{8}(x + 1)^3(4x^2 - 7x + 4),$$ which is greater than zero. }
\end{case}
\begin{case}
    \emph{When $m=3$,  $A(t) = \frac{1}{576}(260 + 273\cos{t} + 28\cos{6t} + 15\cos{7t})$. Again, we can write this in terms of Chebyshev polynomials of the first kind, substituting $T_1(\cos(t))=x$, which leaves
    \begin{align*}
        &\frac{1}{576}(260 + 273x + 28(32x^6 - 48x^4 + 18x^2 - 1) + 15(64x^7 - 112x^5 + 56x^3 - 7x)) =\\ 
        &\frac{1}{72}(x + 1)^3(120x^4 - 248x^3 + 174x^2 - 66x + 29),
    \end{align*}
    which is greater than zero.} 
\end{case}
Thus, $A(t)>0$ on $\left( \frac{(2m-1)\pi}{2m}-\frac{1}{\sqrt{m}}, \pi -\frac{2\pi}{2m+1} \right)$, independently of $m$.

Finally, we show that $A(t)>0$ for $t \in \left( \pi -\frac{2\pi}{2m+1}, \pi \right)$. 
\begin{figure}[H]
    \centering
    \includegraphics[width=0.5\linewidth]{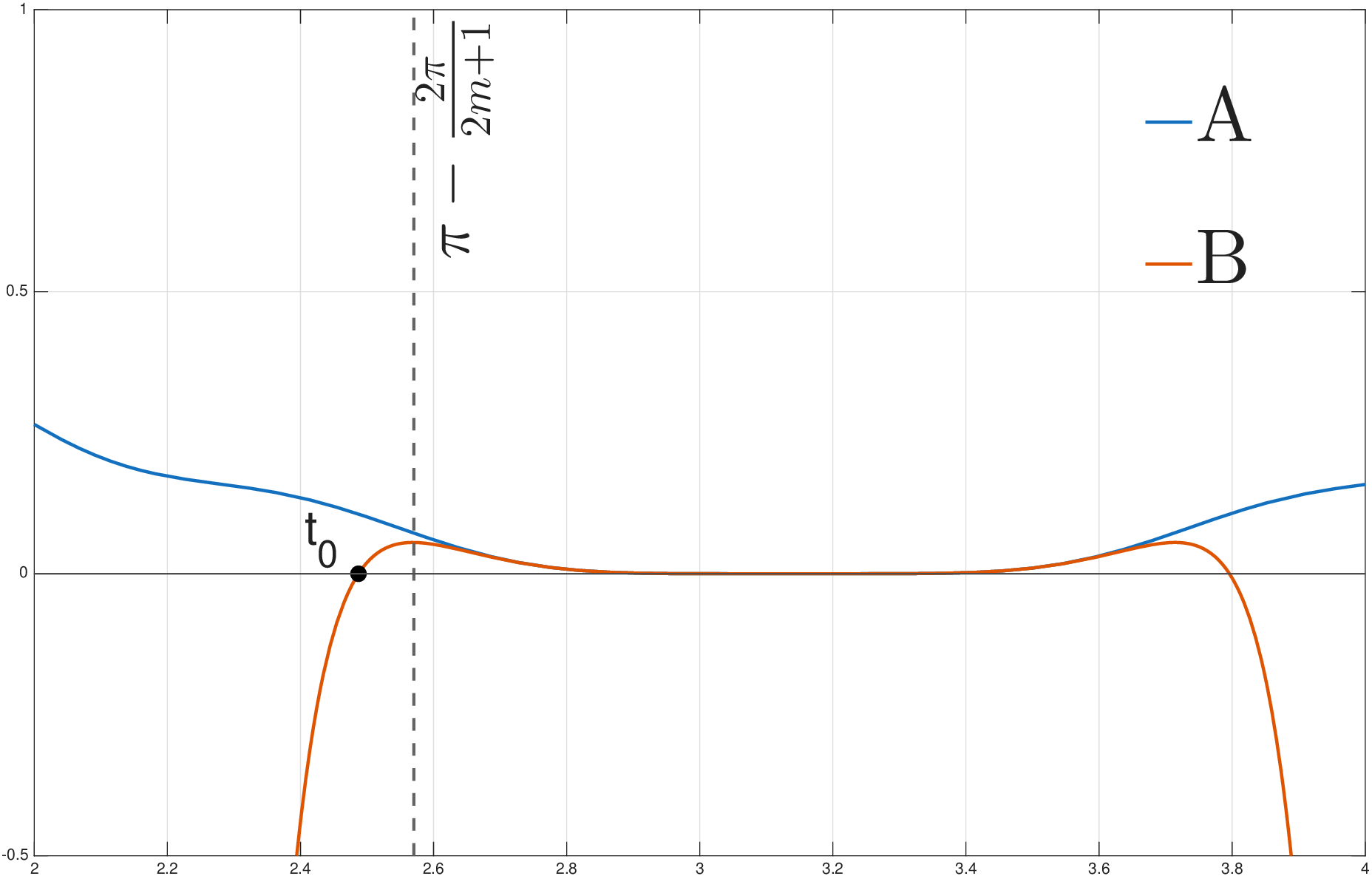}
    \caption{Comparison of $A$ and $B$ for $m=5$.}
    \label{fig:enter-label}
\end{figure}









Let $B(t)$ be the $12$th order Taylor polynomial of $A(t)$ centered at $t = \pi$; i.e.,
\begin{align*}
B(t) &= \frac{A^{(6)}(\pi)}{6!}(t - \pi)^6 + \frac{A^{(8)}(\pi)}{8!}(t - \pi)^8 + \frac{A^{(10)}(\pi)}{10!}(t - \pi)^{10} + \frac{A^{(12)}(\pi)}{12!}(t - \pi)^{12} \\
    &= (t - \pi)^6\left(\frac{A^{(6)}(\pi)}{6!} + \frac{A^{(8)}(\pi)}{8!}(t - \pi)^2 + \frac{A^{(10)}(\pi)}{10!}(t - \pi)^4 + \frac{A^{(12)}(\pi)}{12!}(t - \pi)^6\right).
\end{align*}
Since $A^{(12)}(\pi) < 0$, it follows that $B(t) \leq A(t)$ for every $0 \leq t \leq \pi$. Define the polynomial
\[
R(x) = \frac{A^{(6)}(\pi)}{6!} + \frac{A^{(8)}(\pi)}{8!}x + \frac{A^{(10)}(\pi)}{10!}x^2 + \frac{A^{(12)}(\pi)}{12!}x^3
\]
so that $B(t) = (t - \pi)^6R\big((t - \pi)^2\big)$. 

\par We claim that $R(x)$ has exactly one real root. To see why, notice that its derivative
\[
R'(x) = \frac{A^{(8)}(\pi)}{8!} + \frac{2A^{(10)}(\pi)}{10!}x + \frac{3A^{(12)}(\pi)}{12!}x^2
\]
is a quadratic polynomial with a negative discriminant and a negative leading coefficient. Thus, $R(x)$ is a strictly decreasing cubic polynomial with exactly one real root $x_0 > 0$ since $A(0) = A^{(6)}(\pi)/6! > 0$. This implies that $t_0 = \pi - \sqrt{x_0}$ is the only root of $B(t)$ within the interval $(0, \pi)$. Now as
\begin{align*}
R\left(\left(\frac{2\pi}{2m + 1}\right)^2\right) &= \frac{A^{(6)}(\pi)}{6!} + \frac{(2\pi)^2 A^{(8)}(\pi)}{8!(2m + 1)^2} + \frac{(2\pi)^4A^{(10)}(\pi)}{10!(2m + 1)^4} + \frac{(2\pi)^6A^{(12)}(\pi)}{12!(2m + 1)^6} > 0,
\end{align*}
then it follows that $\big(\frac{2\pi}{2m + 1} \big)^2 < x_0$. Hence $t_0 < \pi - \frac{2\pi}{2m + 1}$, which means that $A(t) > 0$ for all $t \in (0, \pi)$. This completes the proof.
\end{proof}

\subsection{Infinite Family of \texorpdfstring{$7 \times 7$}{7 x 7} Jacobi Matrices with ESE}

\begin{theorem}
\label{Theorem 3.4}
Let $m \geq 1$ be an integer. Then the Jacobi matrix $J$ realizing PST with symmetric spectrum
$
\{0, \pm (2m + 1), \pm (2m + 2), \pm (2m + 3)\}$
has ESE $2m$ times.
\end{theorem}
\vspace{0.25cm}
\begin{proof}
\par Consider the polynomial $P$ defined by $P(\cos{t}) \equiv A(t)$. Then $P$ has a root with multiplicity $3$ at $x = -1$ by Remark \ref{3.1}. Since $\deg{P} = 2m + 3$, it follows from the fundamental theorem of algebra that $P$ has at most $2m$ roots in $(-1, 1)$. Hence $J$ has ESE at most $2m$ times by the bijective correspondence between the roots of $A$ over $(0, \pi)$ and the roots of $P$ over $(-1, 1)$.
\par To show that $J$ has ESE at least $2m$ times, it suffices to check by Bolzano's theorem that the sign of $A$ changes at least $2m$ times over $(0, \pi)$ using the test points $t_k = \frac{k\pi}{2m + 2}$ for every integer $0 \leq k \leq 2m - 1$. Note that the explicit form of $A$ can be obtained from \eqref{2.1} as
\begin{align*}
    A(t) &= \frac{1}{32(m + 1)^2(2m^2 + 4m + 3)}\big(2(4m + 5)(4m + 3)+ (4m + 5)(m + 1)(2m + 3)^2\cos{(2m + 1)t} \\
        &\qquad + 2(4m^2 + 8m + 3)^2\cos{(2m + 2)t} + (2m + 1)^2(4m^2 + 7m + 3)\cos{(2m + 3)t}\big).
\end{align*}
If $k$ is even, then $k = 2l$ for some integer $0 \leq l < m$ and hence
\[
A(t_{2l}) = \frac{1}{2}\left(1 +  \cos \frac{l\pi}{m + 1}\right) > 0.
\]
If instead $k$ is odd, then $k = 2l + 1$ for some integer $0 \leq l < m$ and so
\begin{align*}
A(t_{2l + 1}) &= \frac{-1}{8(m + 1)^2(2m^2 + 4m + 3)}\Bigg(8m^4 + 32m^3 + 36m^2 + 8m - 3 \\
    &\qquad + 4(m + 1)^2(2m^2 + 4m + 3)\cos \frac{(2l + 1)\pi}{2m + 2}\Bigg) \\
    &\leq \frac{-1}{8(m + 1)^2(2m^2 + 4m + 3)}\Bigg(8m^4 + 32m^3 + 36m^2 + 8m - 3 \\
    &\qquad - 4(m + 1)^2(2m^2 + 4m + 3)\cos \frac{3\pi}{2m + 2}\Bigg) \\
    &\leq \frac{-1}{8(m + 1)^2(2m^2 + 4m + 3)}\Bigg(8m^4 + 32m^3 + 36m^2 + 8m - 3 \\
    &\qquad - 4(m + 1)^2(2m^2 + 4m + 3)\left(1 - \frac{1}{2}\Big(\frac{3\pi}{2m + 2}\Big)^2 + \frac{1}{24}\Big(\frac{3\pi}{2m + 2}\Big)^4\right)\Bigg) < 0.
\end{align*}
Thus $A(t_k)$ alternates in parity as $k$ alternates in parity, which implies that $A$ has at least $2m$ roots. Combining this with our last result shows that $J$ has ESE exactly $2m$ times.
\end{proof}

\subsection{Rescaling and Translation of Jacobi Matrices}
It is important to note that both PST and ESE remain invariant under scaling and translation of Jacobi matrices. To see why, notice that \eqref{1.3} implies that $J$ realizes PST if and only if for every $\lambda \in \R$ and $c > 0$, the matrix $cJ + \lambda I$ realizes PST. One may also extend the results from Theorem 1.1 in \cite{EM25} to show that ESE also remains invariant under scaling and translation of Jacobi matrices. Thus we can generalize our results from Theorems \ref{Theorem 3.2} and \ref{Theorem 3.4} to encapsulate \emph{centered} spectra of the form $\{\lambda + c\lambda_k\}_{k = 0}^N$ satisfying \eqref{1.3}.

\section{Observations and Future Work}
Through numerical experiments (see Figures \ref{fig:noese147}-\ref{fig:ese349}), we observed that the ratios between the smallest positive eigenvalue of a $7 \times 7$ Jacobi matrix realizing PST and its other positive eigenvalues determined whether or not $J$ experienced ESE. In particular, we arrived at the following conjecture.
\begin{conjecture}
Let $J$ be a $7 \times 7$ Jacobi matrix with symmetric spectrum realizing PST. Then $J$ does not have ESE if and only if its positive eigenvalues are integer multiples of the smallest positive eigenvalue. 
\end{conjecture}
We believe that the techniques used to prove Theorems \ref{Theorem 3.2} and \ref{Theorem 3.4} cannot be replicated here. 
More generally, we can extend this conjecture for an arbitrary Jacobi matrix of odd order with symmetric spectrum. 

\begin{figure}[H] \label{fig:noese147}
    \centering
    \includegraphics[width=0.35\linewidth]{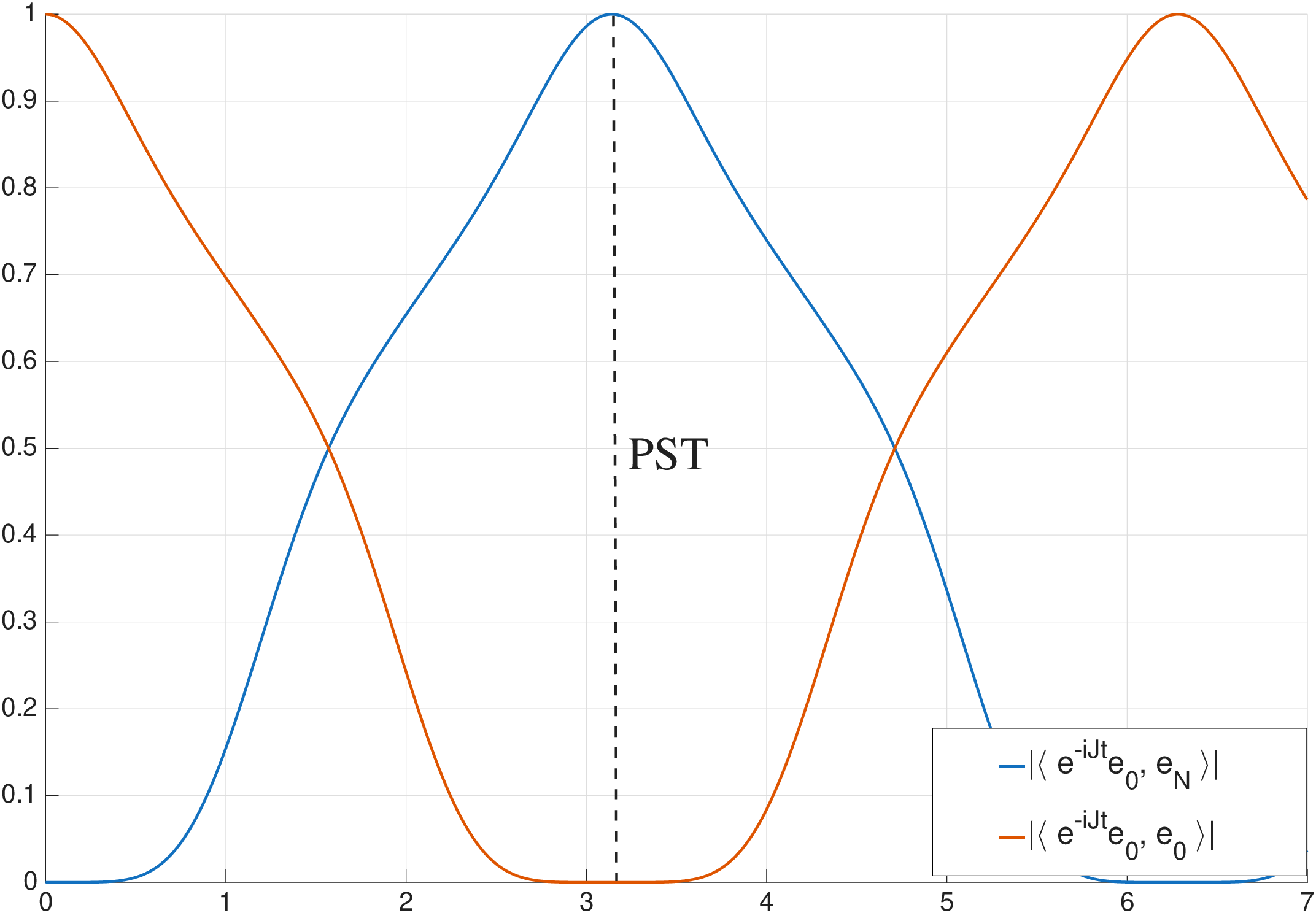}
    \caption{Instance of ESE for $J$ with spectrum $\{0, \pm 1, \pm 4, \pm 7\}$.}
    \label{fig:enter-label}
\end{figure}

\begin{figure}[H]
\label{147}
    \centering
    \includegraphics[width=0.35\linewidth]{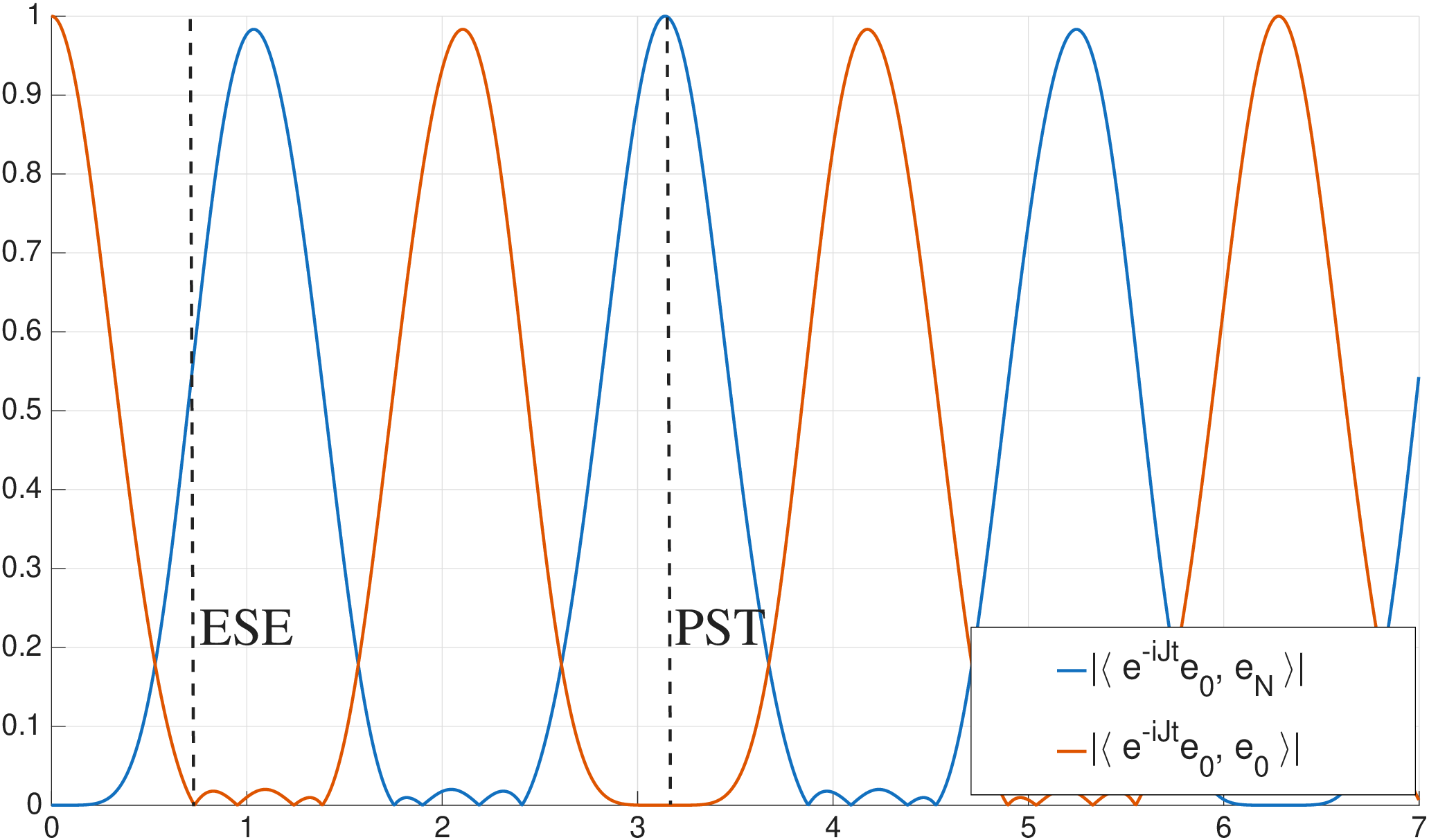}
    \caption{Instance of ESE for $J$ with spectrum $\{0, \pm 3, \pm 6, \pm 11\}$.}
    \label{fig:enter-label}
\end{figure}

\begin{figure}[H]\label{fig:ese349}
\label{3611}
    \centering
    \includegraphics[width=0.35\linewidth]{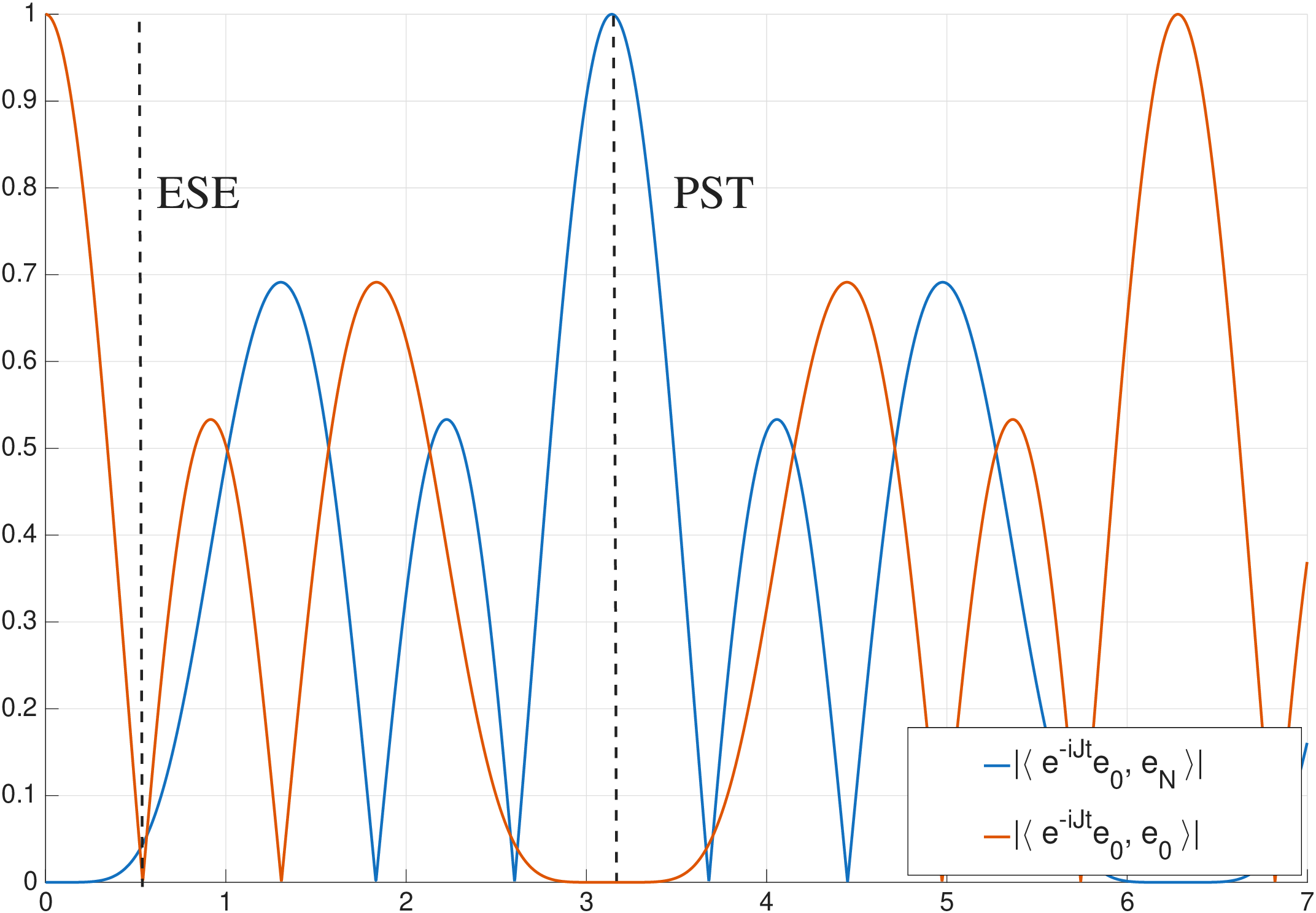}
    \caption{Instance of ESE for $J$ with  spectrum $\{0, \pm 3, \pm 4, \pm 9\}$.}
    \label{fig:enter-label}
\end{figure}

\section*{Acknowledgments}

This work is part of the SIAM-Simons Undergraduate Summer Research Program, which is funded by the Society for Industrial and Applied Mathematics (SIAM) through award 1036702 of the Simons Foundation.
\small

\bibliographystyle{vancouver}
\bibliography{Finaldraft7x7.bib}

\end{document}